\documentclass[11pt]{llncs}

\usepackage{amssymb}
\usepackage{amsmath}
\usepackage{epsfig}
\usepackage{graphics}
\usepackage{color}
\usepackage{llncsdoc}
\usepackage[ruled]{algorithm2e}
\usepackage{fullpage}
\usepackage{pstricks}
\spnewtheorem{observation}[lemma]{Observation}{\bfseries}{\itshape}

\begin{document}

\title{Unified Compression-Based Acceleration\\ of Edit-Distance Computation}

\author{Danny Hermelin\thanks{Supported by the Adams Fellowship of the Israel
Academy of Sciences and Humanities.} \inst{1} Gad M.
Landau\inst{1,2} Shir Landau\inst{3} \and Oren Weimann\inst{4}}

\institute{Department of Computer Science, University of Haifa,
Haifa, Israel. \email{danny@cri.haifa.ac.il,
landau@cs.haifa.ac.il} \and Department of Computer and
Information Science, Polytechnic Institute of NYU, NY, USA.
\and Department of Computer Science, Bar-Ilan University, Ramat
Gan, Israel. \email{shir.landau@live.biu.ac.il} \and  Department of Computer Science and Applied Mathematics, Weizmann Institute of Science, Rehovot, Israel. \email{oren.weimann@weizmann.ac.il} }

\maketitle

\begin{abstract}


The edit distance problem is a classical fundamental problem in
computer science in general, and in combinatorial pattern
matching in particular. The standard dynamic programming
solution for this problem computes the edit-distance between a
pair of strings of total length $O(N)$ in $O(N^2)$ time. To
this date, this quadratic upper-bound has never been
substantially improved for general strings. However, there are
known techniques for breaking this bound in case the strings
are known to compress well under a particular compression
scheme. The basic idea is to first compress the strings, and
then to compute the edit distance between the compressed strings.

\hspace{10pt} As it turns out, practically all known $o(N^2)$
edit-distance algorithms work, in some sense, under the same
paradigm described above. It is therefore natural to ask
whether there is a single edit-distance algorithm that works
for strings which are compressed under any compression scheme.
A rephrasing of this question is to ask whether a single
algorithm can exploit the compressibility properties of strings
under any compression method, even if each string is compressed using a different compression. In this paper we set out to
answer this question by using \emph{straight line programs}. These
provide a generic platform for representing many popular compression schemes
including the LZ-family, Run-Length Encoding, Byte-Pair Encoding, and dictionary methods.

\hspace{10pt}
For two strings of
total length $N$ having straight-line program representations
of total size $n$, we present an algorithm running in $O(nN
\log(N/n))$ time for computing the edit-distance of these two
strings under any rational scoring function, and an
$O(n^{2/3}N^{4/3})$ time algorithm for arbitrary scoring
functions. Our new result, while providing a
significant speed up for highly compressible strings, does not
surpass the quadratic time bound even in the worst case
scenario.

\end{abstract}

\newpage
\pagenumbering{arabic}
\section{Introduction}
\label{Section: Introduction}

Text compression is traditionally applied in order to reduce the use of resources such as storage and bandwidth.
However, in the algorithmic community, there has also been a trend to exploit the properties of compressed data for accelerating the solutions to classical problems on strings. The basic idea is to first compress the input strings, and then
solve the problem on the resulting compressed strings. The compression process in these algorithms is merely a tool, and is used as an intermediate step in the complete algorithm. It is therefore possible that these algorithms may choose to decompress the compressed data after the properties of the compression process have been put to use.   Various
compression schemes, such as LZ77~\cite{ZivLempel1977},
LZW-LZ78~\cite{ZivLempel1976}, Huffman coding, Byte-Pair
Encoding (BPE)~\cite{Shibata-et-al-1999}, and Run-Length Encoding
(RLE), were thus employed to accelerate exact string
matching~\cite{AmirLandauSokol2003,KarkkainenUkkonen1996,Lifshits2007,Manber1994,Shibata-et-al-2000},
subsequence
matching~\cite{CegielskiGuessarianLifshitsMatiyasevich2006},
approximate pattern
matching~\cite{AmirBensonFarach1996,KarkkainenNavarroUkkonen2000,KarkkainenUkkonen1996,NavarroKidaetal2001},
and more~\cite{MozesWeimannZiv2007}.

Determining the \emph{edit-distance} between
a pair of strings is a fundamental problem in computer science
in general, and in combinatorial pattern matching in
particular, with applications ranging from database indexing
and word processing, to bioinformatics~\cite{Gusfield1997}. It
asks to determine the minimum cost of transforming one string
into the other via a sequence of character deletion, insertion,
and replacement operations. Ever since the classical
$O(N^2)$ dynamic programming algorithm by Wagner and
Fisher in 1974 for two input strings of length
$N$~\cite{WagnerFischer1974}, there have been numerous
papers that used compression to accelerate edit-distance
computation. The first paper to break the quadratic time barrier of
edit-distance computation was the seminal paper of Masek and
Paterson~\cite{MasekPaterson1980}, who applied the
"Four Russians technique" to obtain a running time of
$O({N^2}/{\lg N})$ for any pair of strings, and of $O({N^2}/{\lg^2 N})$ assuming a unit cost RAM model. Their algorithm
essentially exploits repetitions in the strings to obtain the
speed up, and so in many ways it can also be viewed as
compression based. In fact, one can say that their algorithm
works on the ``naive compression" that all strings over
constant sized alphabets have.

Apart from its near quadratic runtime, a drawback of the Masek and Paterson algorithm is that it
can only be applied when the given scoring function is
a rational number. That is, when the cost of every elementary character operation is rational. Note that this restriction is indeed a limitation in
computational biology, where PAM and evolutionary distance similarity matrices are used for scoring~\cite{CrochemoreLandauZiv-Ukelson2003,MasekPaterson1980}.
The Masek and Paterson algorithm was later extended to general
alphabets by Bille and Farach-Colton~\cite{BilleFarach2005}.
Bunke and Csirik presented a simple algorithm for computing the
edit-distance of strings that compress well under
RLE~\cite{BunkeCsirik1995}. This algorithm was later improved
in a sequence of
papers~\cite{ApostolicoLandauSkiena1999,ArbellLandauMitchell2001,CrochemoreLandauZiv-Ukelson2003,MakinenNavarroUkkonen1999}
to an algorithm running in time $O(nN)$, for strings of total
length $N$ that encode into run-length strings of total length
$n$. In~\cite{CrochemoreLandauZiv-Ukelson2003}, an algorithm
with the same time complexity was given for strings that are
compressed under LZW-LZ78.
It is interesting to note that all known techniques for
improving on the $O(N^2)$ time bound of edit-distance
computation, essentially apply the acceleration via compression paradigm.

There are two important things to observe from the above:
First, all known techniques for improving on the $O(N^2)$ time
bound of edit-distance computation, essentially apply
acceleration via compression. Second, apart from RLE, LZW-LZ78,
and the naive compression of the Four Russians technique, we do
not know how to efficiently compute edit-distance under other
compression schemes. For example, no algorithm is known which
substantially improves $O(N^2)$ on strings which compress well
under LZ77. Such an algorithm would be interesting since there
are various types of strings that compress much better under
LZ77 than under RLE or LZW-LZ78. In light of this, and due to
the practical and theoretical importance of substantially
improving on the quadratic lower bound of string edit-distance
computation, we set out to answer the following question:
\begin{quote}
``Is there a general compression based edit-distance algorithm
that can exploit the compressibility of two strings under
\emph{any} compression scheme?''
\end{quote}
We propose a unified algorithm for accelerating edit-distance computation
via acceleration. The key idea is to use straight-line programs
(SLPs), which as shown by Rytter~\cite{Rytter2003}, can be used
to model many traditional compression schemes including the
LZ-family, RLE, Byte-Pair Encoding, and dictionary methods. These can
be transformed to straight-line programs quickly and without
large expansion\footnote{Important exceptions of this list are
statistical compressors such as Huffman or arithmetic coding,
as well as compressions that are applied after a
Burrows-Wheeler transformation.}. Thus, devising a fast
edit-distance algorithm for strings that have small SLP
representations, would give a fast algorithm for strings which
compress well under the compression schemes generalized by
SLPs. This has two main advantages:
\begin{enumerate}
\item It allows the algorithm designer to ignore technical
    details of various compression schemes and their associated edit-distance algorithms.
\item One can accelerate edit-distance computation between
    two strings that compress well under \emph{different}
    compression schemes.
\end{enumerate}
In addition, we believe that a fast SLP edit-distance algorithm
might lead to an $O(N^{2-\varepsilon})$ algorithm for general
edit-distance computation, a major open problem in computer
science.

Tiskin~\cite{Tiskin2008} also studied, independently of the
authors, the use of SLPs for edit-distance computation. He gave
an $O(nN^{1.5})$ algorithm for computing the edit-distance
between two SLPs under rational scoring functions. Here, and
throughout the paper, we use $N$ to denote the total length of
the input strings, and $n$ as the total size of their SLP
representations. 
Recently, Tiskin~\cite{Tiskin2009} was able
to speed up his rational scoring function algorithm
of~\cite{Tiskin2008} to an $O(nN \log N)$ algorithm.

\subsection{Our results}
\label{Subsection: Our results}

Initial results for these problems were shown by the authors in ~\cite{HermelinLandauLandauWeimann2009}.
Here, we refine our techniques, allowing us to improve
on all edit-distance computation bounds discussed above. Our
first result of this paper is for the case of rational scoring
functions:
\begin{theorem}
\label{Theorem: First Result}%
Let $\mathcal{A}$ and $\mathcal{B}$ be two SLPs of total size
$n$ that respectively generate two strings $A$ and $B$ of total
length $N$. Then, given $\mathcal{A}$ and $\mathcal{B}$, one can
compute the edit-distance between $A$ and $B$ in $O(nN \log
(N/n))$ time for any rational scoring function.
\end{theorem}

As arbitrary scoring functions are especially important for
biological applications, we obtain the following result for
arbitrary scoring functions:
\begin{theorem}
\label{Theorem: Second Result}%
Let $\mathcal{A}$ and $\mathcal{B}$ be two SLPs of total size
$n$ that respectively generate two strings $A$ and $B$ of total
length $N$. Then, given $\mathcal{A}$ and $\mathcal{B}$, one can
compute the edit-distance between $A$ and $B$ in
$O(n^{2/3}N^{4/3})$ time for any arbitrary scoring function.
\end{theorem}

The reader should compare Theorem~\ref{Theorem: First Result}
to the $O(nN \log N)$ algorithm of Tiskin~\cite{Tiskin2009}, and
Theorem~\ref{Theorem: Second Result} to the $O(n^{1.34}N^{1.34})$ algorithm
of~\cite{HermelinLandauLandauWeimann2009}. In both cases, our
algorithms do not surpass the quadratic bound of $O(N^2)$, even
in the worst case when $n = \Theta(N)$.
There are two main ingredients which we make use of in this paper to
obtain the improvements discussed above:
\begin{enumerate}
\item The recent improvements on $DIST$ merges presented by
    Tiskin~\cite{Tiskin2009}.
\item A refined partitioning of the input strings into
    repeated patterns.
\end{enumerate}
The second ingredient is obtained by much lesser stringent
requirements of the desired partitioning. This has the
advantage that such a partitioning always exists, yet it adds
other technical difficulties which make the version presented
in this sequel more complex. In particular, the construction of
the repository of $DIST$ tables as shown
in~\cite{HermelinLandauLandauWeimann2009} requires a more
careful and detailed analysis.

\subsection{Related Work}
\label{Subsection: Related Work}

Rytter \emph{et al.}~\cite{KarpinskiRytterShinohara1995} were
the first to consider SLPs in the context of pattern matching,
and other subsequent papers also followed this
line~\cite{LehmanShelat2002,MiyazakiShinoharaTakeda1997}.
In~\cite{Rytter2003} and~\cite{Lifshits2007} Rytter and
Lifshits took this work one step further by proposing SLPs as a
general framework for dealing with pattern matching algorithms
that are accelerated via compression. However, the focus of
Lifshits was on determining whether or not these problems are
polynomial in $n$ or not. In particular, he gave an
$O(n^3)$-time algorithm to determine equality of
SLPs~\cite{Lifshits2007}, and he established hardness for the
edit distance~\cite{LifshitsLohrey2006}, and even for the
hamming distance problems~\cite{Lifshits2007}. Nevertheless,
Lifshits posed as an open problem the question of whether or
not there is an $O(nN)$ edit-distance algorithm for SLPs. Here, we attain a bound which is only $\log(N/n)$ away from this.

\subsection{Paper Organization}
\label{Subsection: Paper Organization}

The rest of the paper is organized as follows: In the following
section we present some notation and terminology, and give a
brief discussion of edit-distance computation and SLPs.
Section~\ref{Section: Block Edit-Distance via SLPs} then gives an overview of
the block edit-distance algorithm which is the framework on
which both our algorithms are developed. The two preceding
sections, Sections~\ref{Section: Constructing the x-partition}
and~\ref{Section: Constructing the DIST Repository}, are devoted to
explaining how to take advantage of the SLP representations in
the block edit-distance algorithm. In particular,
Section~\ref{Section: Constructing the x-partition} explains
how to construct a partitioning of the two input strings that
has a very convenient structure, and Section~\ref{Section: Constructing the DIST Repository} explains how to exploit this
structure in order to efficiently construct a repository of
$DIST$ tables to be used in the block edit-distance algorithm.
Finally, in Section~\ref{Section: Putting It All Together}, we complete
all necessary details for proving Theorems~\ref{Theorem: First
Result} and~\ref{Theorem: Second Result}.

\section{Preliminaries}
\label{Section: Preliminaries}

We next present some terminology and notation that will be used
throughout the paper. In particular, we discuss basic concepts
regarding edit-distance computation and straight-line
programs.\\

\noindent \textbf{Edit Distance:} The \emph{edit distance}
between two strings over a fixed alphabet $\Sigma$ is the
minimum cost of transforming one string into the other via a
sequence of character deletion, insertion, and replacement
operations~\cite{WagnerFischer1974}. The cost of these
elementary editing operations is given by some scoring function
which induces a metric on strings over $\Sigma$. The simplest
and most common scoring function is the Levenshtein
distance~\cite{Levenshtein1966} which assigns a uniform cost
of 1 for every operation.

The standard dynamic programming solution for computing the
edit distance between a pair of strings $A=a_1 a_2 \cdots a_N$
and $B=b_1 b_2 \cdots b_N$ involves filling in an $(N+1) \times
(N+1)$ table $T$, with $T[i,j]$ storing the edit distance
between $a_1 a_2 \cdots a_i$ and $b_1 b_2 \cdots b_j$. The
computation is done according to the base case rules given by
$T[0,0] = 0$, $T[i,0] = T[i-1,0] $ + the cost of deleting
$a_i$, and $T[0,j] = T[0,j-1]$  + the cost of inserting $b_j$,
and according to the following dynamic~programming~step\footnote{We note that in most cases, including the Levenshtein
distance~\cite{Levenshtein1966}, when $a_i = b_j$, the cost of replacing $a_i$ with $b_j$ is zero.}:
\begin{equation}
\label{Equation: Edit distance DP}%
T[i,j] = \min
\begin{cases}
T[i-1,j] \text{ + the cost of
deleting $a_i$}\\
T[i,j-1] \text{ + the cost of
inserting $b_j$}\\
T[i-1,j-1] \text{ + the cost of
replacing $a_i$ with $b_j$}\\
\end{cases}
\end{equation}
Note that as $T$ has $(N+1)^2$ entries, the time complexity of
the algorithm above is $O(N^2)$.\\

\noindent \textbf{Straight-line programs:} A
\emph{straight-line program} (SLP) is a context free grammar
generating exactly one string. Moreover, only two types of
productions are allowed: $X_i\rightarrow a$ where $a$ is a
unique terminal, and $X_i\rightarrow X_pX_q$ with \mbox{$i>
p,q$} where $X_1,\ldots, X_n$ are the grammar variables. Each
variable appears exactly once on the left hand side of a
production. In this way, the production rules induce a
rooted ordered tree over the variables of the SLPs, and we can
borrow tree terminology (\emph{e.g.} left-child, ancestor,...)
when speaking about the variables of the SLP. The string
represented by a given SLP is a unique string that is derived by
the last nonterminal $X_n$. We define the size of an SLP to be
$n$, the number of variables it has (which is linear in the
number of productions). The length of the strings that is
generated by the SLP is denoted by $N$. It is important to
observe that SLPs can be exponentially smaller than the
string they generate.

\begin{example}
Consider the string $abaababaabaab$. It could be generated by
the following SLP:\\\\
\medskip
\qquad \qquad \parbox{5cm}{ $X_1\rightarrow b$
 \par $X_2\rightarrow a$
 \par $X_3\rightarrow X_2X_1$
 \par $X_4\rightarrow X_3X_2$
 \par $X_5\rightarrow X_4X_3$
 \par $X_6\rightarrow X_5X_4$
 \par $X_7\rightarrow X_6X_5$
}
 \parbox{8cm}{
\includegraphics[scale=.5]{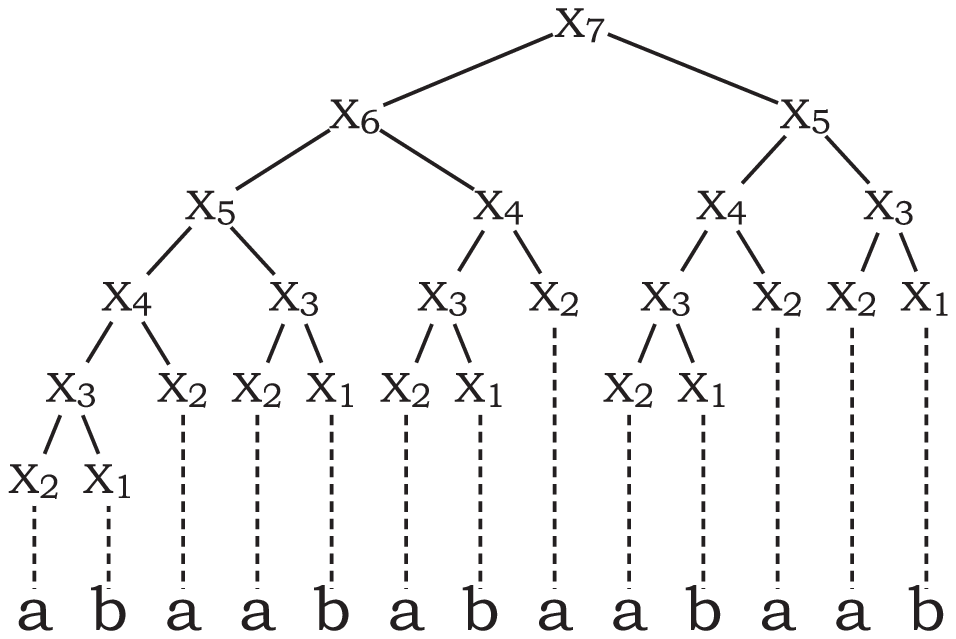}
}
\end{example}

Rytter~\cite{Rytter2003} proved that the resulting encoding of
most compression schemes can be transformed to straight-line
programs quickly and without large expansion. In particular, consider an LZ77
encoding~\cite{LZ77} with $n'$ blocks for a string of length
$N$. Rytter's algorithm produces an SLP-representation with
size $n = O(n' \log N)$ of the same string, in $O(n)$ time.
Moreover, $n$ lies within a $\log N$ factor from the size of a
\emph{minimal} SLP describing the same string. This gives us an
efficient logarithmic approximation of minimal SLPs, since
computing the LZ77 encoding of a string can be done in
linear time. Note also that any string compressed by the
LZ78-LZW encoding~\cite{LZ78} can be transformed directly into a
straight-line program within a constant factor.

\section{The Block Edit-Distance Algorithm}
\label{Section: Block Edit-Distance via SLPs}

In the following section we describe a generic framework for
compression based acceleration of edit distance computation
between two strings called the \emph{block edit-distance}
algorithm. This framework generalizes many
compression based algorithms including the Masek and Paterson
algorithm~\cite{MasekPaterson1980}, and the algorithms
in~\cite{ArbellLandauMitchell2001,CrochemoreLandauZiv-Ukelson2003},
and it will be used for explaining our algorithms in
Theorem~\ref{Theorem: First Result} and Theorem~\ref{Theorem:
Second Result}.

Consider the standard dynamic programming formulation
(\ref{Equation: Edit distance DP}) for computing the
edit-distance between two strings $A=a_1 a_2 \cdots a_N$ and
$B=b_1 b_2 \cdots b_N$. The \emph{dynamic programming grid}
associated with this program, is an acyclic directed graph
which has a vertex for each entry of $T$ (see
Fig.~\ref{Figure: DIST}). The vertex corresponding to $T[i,j]$
is associated with $a_i$ and $b_j$, and has incoming edges
according to (\ref{Equation: Edit distance DP}) -- an edge from
$T[i-1,j]$ whose weight is the cost of deleting $a_i$, an edge
from $T[i,j-1]$ whose weight is the cost of inserting $b_j$,
and an edge from $T[i-1,j-1]$ whose weight is the cost of
replacing $a_i$ with $b_j$. The \emph{value} at the vertex
corresponding to $T[i,j]$ is the value stored in $T[i,j]$,
\emph{i.e.} the edit-distance between the length $i$ prefix of
$A$ and the length $j$ prefix of $B$. Using the
dynamic programming grid $G$, we reduce the problem of
computing the edit-distance between $A$ and $B$ to the problem
of computing the weight of the lightest path from the
upper-left corner to bottom-right corner in $G$.

\begin{figure}[h!]
\begin{center}
\includegraphics[scale=0.6]{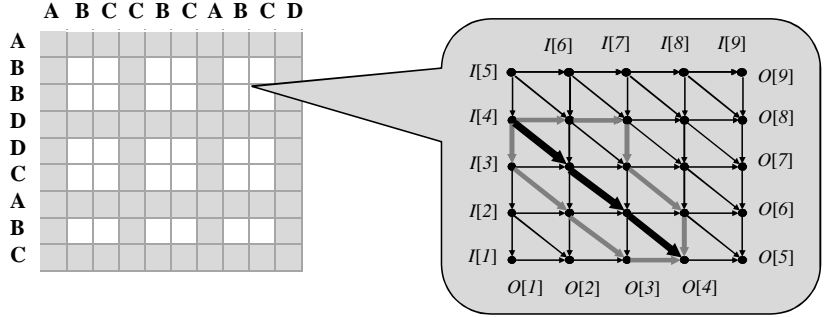}
\end{center}
\caption{A graphical depiction of the Levenshtein distance
dynamic programming grid. On the left, an $x$-partition of the grid.
The boundaries of blocks are shaded in gray. On the right, the $DIST$
table corresponding to the upper-rightmost block in the partition. The
substrings corresponding to this block are ``ABCD" in $A$, and ``ABBD"
in~$B$.}
\label{Figure: DIST}
\end{figure}

We will work with sub-grids of the dynamic programming grid
that will be referred to as \emph{blocks}. The \emph{input
vertices} of a block are all vertices in the first row and
column of the block, while its \emph{output vertices} are all
vertices in the last row and column. Together, the input and
output vertices are referred to as the \emph{boundary} of the
block. The substrings of $A$ and $B$ associated with a block
are defined in the straightforward manner according to its
first row and column. Also, for convenience purposes, we will
order the input and output vertices, with both orderings
starting from the vertex in bottom-leftmost corner of the
block, and ending at the vertex in the upper-rightmost corner.
The $i$th input vertex and $j$th output vertex are the $i$th
and $j$th vertices in these orderings as depicted in Fig.~\ref{Figure: DIST}.

\begin{definition}[\boldmath{$x$}-partition]
Given a positive integer $x \leq N$, an $x$-partition of $G$ is
a partitioning of $G$ into disjoint blocks such that every
block has boundary of size $O(x)$, and there are $O(N/x)$
blocks in each row and column.
\end{definition}

The central dynamic programming tool used by the block
edit-distance algorithm is the $DIST$ table, an elegant and handy
data structure which was originally introduced by Apostolico
\emph{et al.}~\cite{ApostolicoAtallahLarmoreMcFaddin1990}, and
then further developed by others
in~\cite{CrochemoreLandauZiv-Ukelson2003,Schmidt1998} (see
Fig.~\ref{Figure: DIST}).

\begin{definition}[\boldmath{$DIST$}~\cite{ApostolicoAtallahLarmoreMcFaddin1990}]
Let $G'$ be a block in $G$ with $x$ input vertices and $x$
output vertices.  The $O(x^2)$  \emph{\textrm{DIST} table}
corresponding to $G'$ is an $x \times x$ matrix, with
$DIST[i,j]$ storing the weight of the minimum weight path from
the $i$th input to the $j$th output in $G$, and
$\infty$ if no such paths exists.
\end{definition}

In the case of rational scoring
functions, Schmidt~\cite{Schmidt1998} was the first to identify
that the $DIST$ table can be succinctly represented using
$O(x)$ space, at a small cost to query access time. In~\cite{Schmidt1998} the author took advantage of the fact that the number of relevant changes from one column to the next in the $DIST$ matrix is constant. Therefore, the $DIST$ matrix can be fully represented using only the $O(x)$ relevant points, which requires only $O(x)$ space.

\begin{definition}[Succinct \boldmath{$DIST$}~\cite{Schmidt1998}]
A succinct representation of an $x \times x$ $DIST$ table is a
data structure requiring $O(x)$ space, where the value
$DIST[i,j]$, given any $i,j \in \{1,\ldots,x\}$, can be queried
in $O(\log^2 x)$ time.
\end{definition}

It is important to notice that the values at the output
vertices of a block are completely determined by the values at
its input and its corresponding $DIST$ table. In particular, if
$I[i]$ and $O[j]$ are the values at the $i$th input vertex and
$j$th output vertex of a block $G'$ of $G$, then
\begin{equation}
\label{Equation: Block Output}%
O[j] = \min_{1 \leq i \leq x} (I[i] + DIST[i,j]).
\end{equation}
By (\ref{Equation: Block Output}), the values at the output
vertices of $G'$ are given by the column minima of the matrix
$I + DIST$. Furthermore, by a simple modification of all
$\infty$ values in $I + DIST$, we get what is known as a
\emph{totally monotone
matrix}~\cite{CrochemoreLandauZiv-Ukelson2003}. Aggarwal
\emph{et al.}~\cite{AggarwalKlaweMoranShorWilber1987} gave an elegant recursive algorithm, nicknamed SMAWK in the literature,
that computes all column minima of an $x \times x$
totally monotone matrix by querying only $O(x)$ elements of the
matrix. It follows that using SMAWK we can compute the output
values of $G'$ in $O(x)$ time. We will be using this mechanism
for the arbitrary scoring function. However, for the case of rational scoring functions, using SMAWK on the
succinct representation of $DIST$ requires $O(x\log^2 x)$ time.
Tiskin~\cite{Tiskin2009} showed how to reduce this to $O(x\log x)$ using a simple
divide-and-conquer approach.

In addition, let us now discuss how to efficiently construct the $DIST$ table corresponding to a block in $G$. Observe that this can be done quite easily in $O(x^3)$ time, for blocks with boundary size $O(x)$, by computing the standard dynamic programming table between every prefix of $A$ against $B$ and every prefix of $B$ against $A$. Each of these dynamic programming tables contains all values of a particular row in the $DIST$ table. In~\cite{ApostolicoAtallahLarmoreMcFaddin1990}, Apostolico \emph{et al.} show an elegant way to reduce the time complexity of this construction to $O(x^2 \lg x)$. In the case of rational scoring functions, the complexity can be further reduced to $O(x^2)$ as shown by Schmidt~\cite{Schmidt1998}.\\\\

\noindent \textbf{Block Edit Distance}
\begin{enumerate}
\item Construct an $x$-partition of the dynamic programming
    grid of $A$ and $B$, with some $x \leq N$ to be chosen
    later.
\item Construct a repository with the $DIST$ tables
    corresponding to each block in the $x$-partition.
\item Fill in the first row and column of $G$ using the
    standard base case rules.
\item In top-to-bottom and left-to-right manner, identify
    the next block in the partition of $G$ and use its
    input and the repository to compute its output using
    (\ref{Equation: Block Output}). Use the outputs in
    order to compute the inputs of the next block using
    (\ref{Equation: Edit distance DP}).
\item The value in the bottom-rightmost cell is the edit
    distance of $A$ and $B$.
\end{enumerate}

The first two steps of the block edit-distance algorithm are
where we actually exploit the repetitive structure of the input
strings induced by their SLP representations. These will be
explained in further detail in Sections~\ref{Section:
Constructing the x-partition} and~\ref{Section: Constructing
the DIST Repository}. Apart from these two steps, all details
necessary for implementing the block edit-distance algorithm
should by now be clear. Indeed, steps 3 and 5 are trivial, and
step 4 is computed differently for rational scoring schemes and
arbitrary scoring schemes according to the above discussion.
Since there are $O(N^2/x^2)$ blocks induced by an $x$-partition,
this step requires $O(N^2/x)$ time for arbitrary scoring
functions, and $O(N^2 \lg x /x)$ time when dealing with
succinct $DIST$ tables in the case of rational scoring
functions.


\section{Constructing the $\boldmath{x}$-partition}
\label{Section: Constructing the x-partition}

In this section we explain how to construct an $x$-partition of the dynamic programming grid $G$. Each block in this partition, in addition to being associated with a substring $A'$ of $A$ and a substring $B'$ of $B$, is also associated with unique grammar variables $v\in\mathcal{A}$ and
$u\in\mathcal{B}$. 
It is possible that $A'$ is only a prefix or a suffix of the string that is derived from $v$. Similarly, it is possible that $B'$ is only a prefix or a suffix of the string that is derived from $u$. However, every pair of grammar variables $(v\in\mathcal{A}, u\in\mathcal{B})$ is only associated with  one pair of substrings $(A',B')$. Notice that in the $x$-partition, it is possible that more than one block is associated with substrings $(A',B')$. This is due to the inherent repetitions in the parse tree that are crucial for the efficiency of our algorithm.

\begin{figure}[t]
\begin{center}
\parbox{5cm}{\includegraphics[scale=0.65]{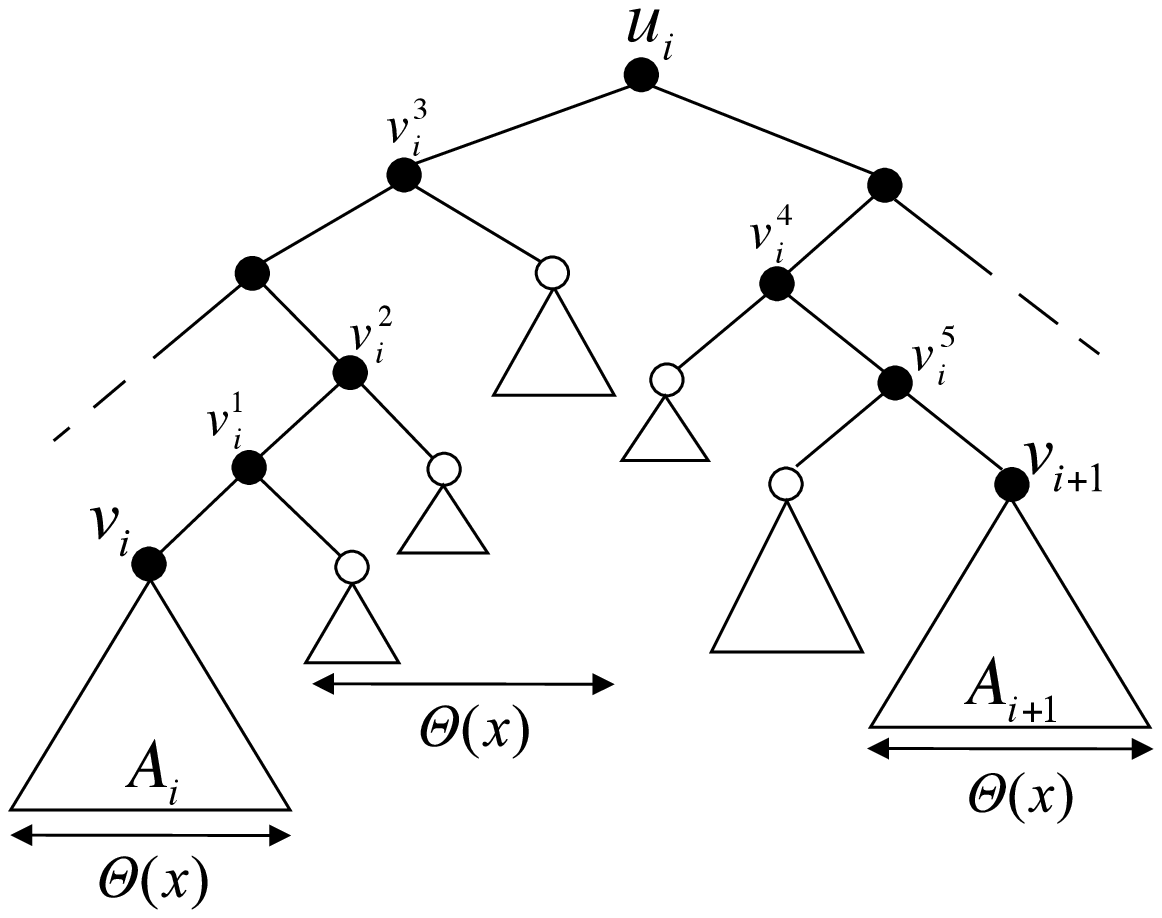}}
\parbox{3cm}{\ \ \   }
\parbox{12cm}{The key-vertices
$v_i$ and $v_{i+1}$ both derive strings of length $\Theta(x)$, and their least common ancestor is $u_i$. The white vertices ``hanging'' off the $v_i$-to-$v_{i+1}$ path are variables that together derive the substrings yet to be covered. Notice that the white vertices derive strings of length shorter than $x$.
In the final partition, $v_i^3$ is associated with an actual substring in the partition. $v_i^2$ is associated with a prefix of this substring.
}
\end{center}
\caption{A closer look on the parse tree of an SLP $\mathcal{A}$.}
 \label{Figure: keyvertices}
\end{figure}

To describe the $x$-partition procedure, we explain how $A$ is partitioned into substrings, $B$ is partitioned similarly.
We show that for every SLP $\mathcal{A}$
generating a string $A$ and every $x\le N$, one can partition
$A$ into $O(N/x)$ disjoint substrings, each of length $O(x)$,
such that every substring is the complete or partial generation of some variable in
$\mathcal{A}$. The outline of the partition process is as follows: 
\begin{enumerate}
\item We first identify the grammar variables in
    $\mathcal{A}$ which generate a disjoint substring of length
    between $x$ and $2x$. There are at most $O(N/x)$ such variables.
\item We then show that the substrings of $A$ that are still not
    associated with a variable can each be generated by $O(N/x)$
    additional variables. Furthermore, while, these additional $O(N/x)$ variables may generate a string of length greater than $x$, we will show how to extract only the desired substring from the string that they generate. We add all such
    variables to our partition of $A$ for a total of $O(N/x)$ variables. 
\end{enumerate}
    
To clarify, lets look at the example depicted in Fig.~\ref{Figure: keyvertices}. In step 1 described above, the vertexes $v_i$ and $v_{i+1}$ are selected. In step 2, vertices $v_i^3$ and $v_i^4$ are selected. For the latter, we will use only a portion of the strings generated by these variables as needed by the algorithm.

We now give a detailed description of the partition process. To partition $A$, consider the parse tree of $\mathcal{A}$ as depicted in Fig.~\ref{Figure: keyvertices}.
We begin by assigning every grammar variable (vertex) that derives a string shorter than $x$ with the exact string that it derives.
We continue by marking every vertex $v$ that derives a string
of length greater than $x$ as a {\em key-vertex} iff both children of
$v$ derive strings of length smaller than $x$. This gives
us $ O(N/x)$ key-vertices, such that each derives a
string of length $\Theta(x)$. We then partition $A$ according to these vertices. In particular, we associate each key-vertex with the exact string that it derives.
But we are still not guaranteed that the key-vertices cover $A$ entirely.

To fix this, we take a look at two key-vertices $v_i,v_{i+1}$
selected in the above process as seen in Fig.~\ref{Figure:
keyvertices}. Assume $v_i$ derives the string $A_i$ and
$v_{i+1}$ derives the string $A_{i+1}$, and
that $A'_i$ is the ``missing'' substring of $A$ that lies
between $A_{i}$ and $A_{i+1}$. Note that both $A_i$ and
$A_{i+1}$ are of length $\Theta(x)$, $A'_i$ however, can be
either longer than $\Theta(x)$ or shorter.
We now go on to show
how to partition $A'_i$ into $O(|A'_i|/x)$ substrings of length $O(x)$ each.

Let $u_i$ be the lowest common ancestor of $v_i$ and $v_{i+1}$, and let
$v_{i}^1, \ldots,  v_{i}^s$ (resp. $v_{i}^{s+1}, \ldots, v_{i}^{t}$) be the vertices, not including $u_i$, on the unique $v_i$-to-$u_i$ (resp. $u_i$-to-$v_{i+1}$) path whose right (resp. left) child is also on the path. We focus on partitioning the substring of $A'_i$ corresponding to the $v_i$-to-$u_i$ path (partitioning the $u_i$-to-$v_{i+1}$ part is done symmetrically). The following procedure partitions this part of $A'_i$ and associates every one of $v_{i}^1,v_{i}^2,\ldots,v_{i}^s$ with a substring.

\begin{enumerate}
\item initialize $j=1$
\item while $j\le s$
\item \ \ \ \ \ associate $v_i^j$ with the string derived by its right child, and initialize $S$ to be this substring
\item \ \ \ \ \ while $|S| < x$ and $j\neq s$
\item  \ \ \ \ \ \ \ \ \ \ \ \  concatenate the string derived by  $v_i^{j+1}$'s  right child to $S$
\item  \ \ \ \ \ \ \ \ \ \ \ \  associate the new $S$ with $v_i^{j+1}$
\item  \ \ \ \ \ \ \ \ \ \ \ \  $j \leftarrow j+1$
\item  \ \ \ \ \ set $S$ as a string in the $x$-partition
\end{enumerate}

It is easy to verify that the above procedure partitions  $A'_i$ into $O( |A'_i|/x$) substrings, where one substring can be shorter than $x$ and all the others are of length between $x$ and $2x$. Therefore, after applying this procedure to all  $A'_i$s, $A$ is partitioned into $O(N/x)$ substrings each of length $O(x)$.
It is also easy to see that we can identify the key-vertices, as well as the $v_i^j$ vertices in $O(N)$ time.

An important observation that we point out is that the basic structure of an SLP grammar constitutes that every internal node in the tree represents a variable in the grammar and a grammar variable of the form $v_i^j$ is always associated with the {\em same} substring $S$ (and $|S|\le 2x$). Due to the bottom-up nature of the above process, the same respective key-vertices in the subtree of a given variable, will be chosen for any appearance of that variable in the tree. This is because in every place in the parse tree where $v_i^j$ appears, the subtree rooted at $v_i^j$ is exactly the same, so the above (bottom-up) procedure would behave the same.

%

\section{Constructing the $\boldmath{DIST}$ Repository}
\label{Section: Constructing the DIST Repository}

In the previous section, we have
seen how to construct an $x$-partition of the
dynamic programming grid $G$. Once this partition has been
built, the first step of the block edit-distance procedure is
to construct a repository of $DIST$ tables corresponding to
each block of the partition. In this section we discuss how to
construct this repository of $DIST$ tables efficiently.

We will be building the $DIST$ tables by utilizing the process
of merging two $DIST$ tables. That is, if $D_1$ and $D_2$ are
two $DIST$ tables, one between a pair of substrings $A'$ and
$B'$ and the other between $A'$ and $B''$, then we refer to the
$DIST$ table between $A'$ and $B'B''$ as the product of
\emph{merging} $D_1$ and $D_2$. For an arbitrary scoring
function, merging two $x\times x$ $DIST$ tables requires $O(x^2)$
time using $x$ iterations of the SMAWK algorithm discussed in
Section~\ref{Section: Block Edit-Distance via SLPs}. For the
rational scoring function, a recent important result of
Tiskin~\cite{Tiskin2009} shows how using the succinct
representation of $DIST$ tables, two $DIST$ tables can be
merged in $O(x \lg x)$ time.

Recall that the $x$-partitioning step described in
Section~\ref{Section: Constructing the x-partition}
associates SLP variables with substrings in $A$ and $B$. Each
substring which is associated with a variable is of length
$O(x)$. There are two types of associated substrings: Those
whose associated SLP variable derives them exactly, and those
whose associated variable derives a superstring of them. As a first step, we construct the
$DIST$ tables between any pair of substrings $A'$ and $B'$ in
$A$ and $B$ that are associated with a pair of SLP variables,
and are of the first type.
This is done in a recursive manner. We first construct
the  four $DIST$ tables that correspond to
the pairs $(A'_1,B'_1)$, $(A'_1,B'_2)$, $(A'_2,B'_1)$, and
$(A'_2,B'_1)$, where $A'_1$ and $A'_2$ (resp. $B'_1$ and $B'_2$) are the strings derived by the left and right
children of $A'$'s (resp. $B'$'s) associated variable.
We then merge these four $DIST$s together
to obtain a $DIST$ between $A'$ and $B'$.

Now assume $A'$ is a substring in $A$ of the second type, and
we want to construct the $DIST$ between $A'$ and a substring
$B'$ in $B$. For simplicity assume that $B'$ is a substring of
the first type. Then, the variable associated with $A'$ is a
variable of the form $v^j_i$ on a path between some two key
vertices $v_i$ and $v_{i+1}$.
For each $k \leq j$,
let $A^k_i$ denote the substring associated with  $v^k_i$. Notice that $A'$ is the concatenation of $A^{j-1}_i$ and the substring $A''$ that is derived from $v^j_i$'s right child.
The $DIST$
between $A'$ and $B'$ is thus constructed by first recursively
constructing the $DIST$ between $A^{j-1}_i$ and $B'$,
and then merging this with the $DIST$ between $A''$ and $B'$.
The latter already being available since the variables associated with $A''$ and $B'$ are of the first type.

\begin{lemma}
\label{Lemma: dist repository rational}%
Building the $DIST$ repository under a rational scoring function
can be done in $O(n^2x\lg x)$ time.
\end{lemma}
\begin{proof}
Recall that for rational scoring functions, merging two
succinct $DIST$s of size $O(x)$ each requires $O(x \lg x)$
time~\cite{Tiskin2009}. We perform exactly one merge per each
distinct pair of substrings in $A$ and $B$ induced by the
$x$-partition. Since each substring is associated with a unique
SLP variable, there can only be $O(n^2)$ distinct
substring pairs. Thus, we get that only $O(n^2)$ merges of
succinct $DIST$s need to be performed to create our
repository, hence we achieve the required bound. \qed
\end{proof}

\begin{lemma}
\label{Lemma: dist repository arbitrary}%
Building the $DIST$ repository under an arbitrary scoring
function can be done in $O(n^2x^2)$ time.
\end{lemma}
\begin{proof}
In the case of arbitrary scoring functions, merging two
$x\times x$ $DIST$ tables requires $O(x^2)$ time. The same
upper bound shown above of $O(n^2)$ merges may be performed in
this case as well, and therefore we achieve the required bound.
\qed
\end{proof}

\section{Putting It All Together}
\label{Section: Putting It All Together}

As the major components of our algorithms have now been
explained, we go on to summarize our main results. In
particular, we complete the proof of Theorem~\ref{Theorem:
First Result} and~\ref{Theorem: Second Result}.

\begin{lemma}
\label{Lemma: summary rational}%
The block edit distance algorithm runs in $O(n^2x \lg x + N^2
\lg x/x )$ time in case the underlying scoring function is
rational.
\end{lemma}

\begin{proof}
The time required by the block edit-distance algorithm is
dominated by computing the repository of $DIST$ tables in step
2, and the cost of computing the output values of each block in
step 4. Steps 1,3 and 5 all take linear time or less. As we
have shown in Section~\ref{Section: Constructing the DIST
Repository}, in case the underlying scoring function is
rational, the $DIST$ table repository can be built in
$O(n^2x\lg x)$ time, accounting for the first component of our
bound. As explained in Section~\ref{Section: Block
Edit-Distance via SLPs}, step 4 of the algorithm requires
$O(N^2 \lg x/x )$ time, accounting for the second component in
the above bound. \qed
\end{proof}

To conclude, constructing an $x$-partition with $x = N/n$, we
get a time complexity of $O(nN \lg (N/n))$, proving
Theorem~\ref{Theorem: First Result}.

\begin{lemma}
\label{Lemma: summary arbitrary}%
The block edit distance algorithm runs in $O(n^2x^2 + N^2/x)$
time in case the underlying scoring function is arbitrary.
\end{lemma}

\begin{proof}
As described above, the time required by our algorithm is
dominated by computing the repository of $DIST$ tables in step
2, and the cost of propagating the dynamic programming values
through block boundaries in step 4. As we have shown in
Section~\ref{Section: Constructing the DIST Repository}, in
case the underlying scoring function is arbitrary, the
repository of $DIST$ tables can be built in $O(n^2x^2)$ time,
accounting for the first component of our bound. As explained
in Section~\ref{Section: Block Edit-Distance via SLPs}, step 3
requires $O(N^2/x)$ time, therefore completing the proof of the
above bound. \qed
\end{proof}

To conclude, constructing an $x$-partition with $x=(N/n)^{2/3}$,
we get a time complexity of $O(n^{2/3}N^{4/3})$, proving
Theorem~\ref{Theorem: Second Result}.

\bibliographystyle{plain}
\bibliography{biblio}

\end{document}